\documentclass[letterpaper, 10 pt, conference]{ieeeconf}  

\IEEEoverridecommandlockouts                              
\usepackage{fncylab,enumerate,wrapfig,placeins}
\usepackage{bbm}
\usepackage{graphicx}
\usepackage{amsmath,amssymb}  
\usepackage{geometry}
\usepackage{eufrak}
\usepackage{multicol}

\pdfoutput=1  

\usepackage{subcaption}
\usepackage{booktabs}

\usepackage{multirow}

\usepackage{algorithm}
\usepackage{algpseudocode}

\usepackage{algorithm}
\usepackage{algpseudocode}
\algnewcommand\algorithmicinput{\textbf{Input:}}
\algnewcommand\INPUT{\item[\algorithmicinput]}
\algnewcommand\algorithmicoutput{\textbf{Output:}}
\algnewcommand\OUTPUT{\item[\algorithmicoutput]}

\usepackage[usenames,dvipsnames]{color}
\usepackage{color}

\usepackage[colorlinks,%
linkcolor=BrickRed,%
filecolor =BrickRed,%
citecolor=RoyalPurple,%
]{hyperref}

\newlength{\noteWidth}
\long\def\notes#1{\ifinner
{\footnotesize #1}
\else 
\marginpar{\parbox[t]{\noteWidth}{\raggedright\footnotesize#1}}
\fi\typeout{#1}}

\def\notes#1{}

\def\spm#1{\notes{SPM:  #1}} 
\def\ad#1{\notes{AD:  #1}}

\def\archive#1{\typeout{save this note:  #1}}

\def\rd#1{{\color{red}#1}} 
 
\def\mindex#1{\index{#1}}

\def\sq{\hbox{\rlap{$\sqcap$}$\sqcup$}}
\def\qed{\ifmmode\sq\else{\unskip\nobreak\hfil
\penalty50\hskip1em\null\nobreak\hfil\sq
\parfillskip=0pt\finalhyphendemerits=0\endgraf}\fi\medskip}

\long\def\defbox#1{\framebox[.9\hsize][c]{\parbox{.85\hsize}{%
\parindent=0pt
\baselineskip=12pt plus .1pt      
\parskip=6pt plus 1.5pt minus 1pt 
 #1}}}

\long\def\beginbox#1\endbox{\subsection*{}%
\hbox{\hspace{.05\hsize}\defbox{\medskip#1\bigskip}}%
\subsection*{}}

\def\endbox{}

\def\transpose{{\hbox{\it\tiny T}}}

\newsavebox{\junk}
\savebox{\junk}[1.6mm]{\hbox{$|\!|\!|$}}

\def\det{{\mathop{\rm det}}}

\def\state{{\sf X}}

\newcommand{\field}[1]{\mathbb{#1}}

\def\Re{\field{R}}

\def\ind{\field{I}}

\def\bfmath#1{{\mathchoice{\mbox{\boldmath$#1$}}%
{\mbox{\boldmath$#1$}}%
{\mbox{\boldmath$\scriptstyle#1$}}%
{\mbox{\boldmath$\scriptscriptstyle#1$}}}}

\def\bfmb{\bfmath{b}}

\def\bfmw{\bfmath{w}}

\def\bfmX{\bfmath{X}}

\def\bfmY{\bfmath{Y}}

\def\bfmhhaY{\bfmath{\hhaY}} 
\def\bfmhhaY{\hbox to 0pt{$\widehat{\bfmY}$\hss}\widehat{\phantom{\raise 1.25pt\hbox{$\bfmY$}}}}

\def\bfmZ{\bfmath{Z}}

\def\haf{{\hat f}}

\def\til={{\widetilde =}}

\def\tiltheta{\widetilde \theta}

\def\tiltheta{{\tilde \theta}}

\def\clB{{\cal B}}

\def\clE{{\cal E}}
\def\clF{{\cal F}}

\def\clS{{\cal S}}
\def\clT{{\cal T}}

 \def\FRAC#1#2#3{\genfrac{}{}{}{#1}{#2}{#3}}

\def\ddt{{\mathchoice{\FRAC{1}{d}{dt}}%
{\FRAC{1}{d}{dt}}%
{\FRAC{3}{d}{dt}}%
{\FRAC{3}{d}{dt}}}}

\def\ddtp{{\mathchoice{\FRAC{1}{d^{\hbox to 2pt{\rm\tiny +\hss}}}{dt}}%
{\FRAC{1}{d^{\hbox to 2pt{\rm\tiny +\hss}}}{dt}}%
{\FRAC{3}{d^{\hbox to 2pt{\rm\tiny +\hss}}}{dt}}%
{\FRAC{3}{d^{\hbox to 2pt{\rm\tiny +\hss}}}{dt}}}}

\def\half{{\mathchoice{\FRAC{1}{1}{2}}%
{\FRAC{1}{1}{2}}%
{\FRAC{3}{1}{2}}%
{\FRAC{3}{1}{2}}}}

\def\eqdef{\mathbin{:=}}

\def\Expect{{\sf E}}

\def\average#1,#2,{{1\over #2} \sum_{#1}^{#2}}

\def\eye(#1){{\bf(#1)}\quad}
\def\epsy{\varepsilon}

\newtheorem{theorem}{Theorem}[section]

\newtheorem{proposition}[theorem]{Proposition}
\newtheorem{lemma}[theorem]{Lemma}

\def\Lemma#1{Lemma~\ref{#1}}

\def\Theorem#1{Theorem~\ref{#1}}

\def\Section#1{Section~\ref{#1}}

\def\eq#1/{(\ref{e:#1})}

\newcommand{\beqn}[1]{\notes{#1}%
\begin{eqnarray} \elabel{#1}}

\newcommand{\eeqn}{\end{eqnarray} }

\newcommand{\beq}[1]{\notes{#1}%
\begin{equation}\elabel{#1}}

\newcommand{\eeq}{\end{equation}}

\def\bdes{\begin{description}}
\def\edes{\end{description}}

\newcounter{rmnum}
\newenvironment{romannum}{\begin{list}{{\upshape (\roman{rmnum})}}{\usecounter{rmnum}
\setlength{\leftmargin}{14pt}
\setlength{\rightmargin}{8pt}
\setlength{\itemsep}{2pt}
\setlength{\itemindent}{-1pt}
}}{\end{list}}

\newcounter{anum}

{\end{list}}

\def\ass(#1:#2){(#1\ref{#1:#2})}

\def\ritem#1{
\item[{\sf \ass(\current_model:#1)}]
}

\newenvironment{recall-ass}[1]{%
\begin{description}
\def\current_model{#1}}{
\end{description}
}

\def\Ebox#1#2{%
\begin{center}
 \parbox{#1\hsize}{\epsfxsize=\hsize \epsfbox{#2}}
\end{center}}

\newcommand{\bd}{\begin{description}}
\newcommand{\ed}{\end{description}}
\newcommand{\bt}{\begin{theorem}}
\newcommand{\et}{\end{theorem}}
\newcommand{\ba}{\begin{array}{rcl}}
\newcommand{\ea}{\end{array}}

\def\assume#1{\smallbreak\noindent\textbf{#1}}

\def\td{d}

\def\Real{\text{Re}}

\def\state{{\sf X}}

\def\eqdef{\mathbin{:=}}

\def\Lemma#1{Lemma~\ref{#1}}

\def\Prop#1{Prop.~\ref{#1}}
\def\Theorem#1{Thm.~\ref{#1}}   

\def\haA{\widehat A}

\def\tilb{\tilde{b}}

\def\tilp{\widetilde{p}}

\def\tilA{\tilde{A}}

\def\interest{\Gamma}

\graphicspath{{SFQFigures/}}

\def\Ebox#1#2{%
\begin{center}
\includegraphics[width= #1\hsize]{#2} 
\end{center}}

\def\Fig#1{Fig.~\ref{#1}}

\def\ind{\field{I}}

\def\Re{\field{R}}

\def\barw{\overline{w}}

\def\barb{\bar b}

\def\cstop{{c}_{{s}}}
\def\barcstop{\overline{c}_{{s}}}

\def\barclA{\bar{\cal A}}

\title{Zap~Q-Learning for Optimal Stopping Time Problems}

\author{%
Shuhang Chen\authorrefmark{1} 
\and
Adithya M. Devraj\authorrefmark{2} 
\and
Ana Bu\v{s}i\'{c}\authorrefmark{3} 
 \and 
 Sean  Meyn\authorrefmark{4}%
\thanks{\authorrefmark{1}S.C.\ is with the Department of Mathematics at the University of Florida in Gainesville}%
\thanks{\authorrefmark{2}A.D.\ is with the Department of ECE at the University of Florida}%
\thanks{\authorrefmark{3}A.B.\ is with Inria and DI ENS, \'Ecole Normale	Sup\'erieure, CNRS, PSL Research University, Paris, France}%
\thanks{\authorrefmark{4}S.M.\ is with Department of Electrical and Computer Engineering, University of Florida, and Inria International Chair, Paris}%
\thanks{\textbf{Acknowledgements}:
 Financial support from ARO grant W911NF1810334 is gratefully acknowledged.  Additional support from EPCN 1609131 \&\ CPS~1646229, and French National Research Agency grant ANR-16-CE05-0008.
  }%
  }
  
\begin{document}

\maketitle

\begin{abstract}   
	
	The objective in this paper is to obtain fast converging reinforcement learning algorithms to approximate solutions to the problem of discounted cost optimal stopping in an irreducible, uniformly ergodic Markov chain, evolving on a compact subset of $\Re^n$. We build on the dynamic programming approach taken by Tsitsikilis and Van Roy, wherein they propose a Q-learning algorithm to estimate the optimal state-action value function, which then defines an optimal stopping rule. We provide insights as to  why the convergence rate of this algorithm can be slow, and propose a fast-converging alternative, the ``Zap-Q-learning" algorithm, designed to achieve {optimal rate of convergence}. For the first time, we prove the convergence of the Zap-Q-learning algorithm under the assumption of linear function approximation setting. We use ODE analysis for the proof, and the optimal asymptotic variance property of the algorithm is reflected via fast convergence in a finance example.
	
\end{abstract}

\def\todo{{\color{red}To do:}\ }


\section{Introduction}
\label{s:intro}

Consider a discrete-time Markov chain $\bfmX = \{X_n: n \geq 0 \}$ evolving on a general state-space $\state $.
The goal in optimal stopping time problems is to minimize over all stopping times $\tau$, the   associated expected cost:
\begin{equation}
\Expect \Bigg[ \sum_{n = 0}^{\tau-1} \beta^n c(X_n) + \beta^\tau \cstop(X_\tau) \Bigg] 
\label{e:opt_stop}
\end{equation}
where $c: \state \to \Re$ denotes the per-stage cost, $\cstop : \state \to \Re$   the terminal cost, and $\beta \in(0,1)$ is the discount factor. Examples of such problems arise mostly in financial applications such as derivatives analysis (see \Section{s:num}), timing of a purchase or sale of an asset, and more generally in problems that involve sequential analysis.

In this work, the optimal decision rule is approximated using reinforcement learning techniques.  
We propose and analyze an optimal  variance algorithm to   approximate the value function associated with the optimal stopping rule.


\subsection{Definitions \& Problem Setup}

We assume that the state-space $\state \subset \Re^m$ is compact,  and we let $\clB$ denote the associated Borel $\sigma$-algebra.
The time-homogeneous Markov chain $\bfmX$ is defined on a probability space ($\Omega,\mathcal{F},\mathcal{P}$), and its {dynamics} is determined by an initial distribution {$\mu: \state \to [0, 1]$}, and a transition kernel $P$: for each $x\in\state$ and $A\in\clB$,
\[
P(x, A) =
\text{Pr}(X_{n+1} \in A \mid X_n = x ) 
\]
It is assumed that $\bfmX$ is \emph{uniformly ergodic}: There exisits a unique invariant probability measure $\pi$, a constant $D < \infty$, and $0 < \rho < 1$, such that, for all $x \in \state$ and $A \in \clB$,
\begin{equation}
\| P^{n} (x, A) - \pi(A) \| \leq D \rho^n \,, \qquad n \geq 0  
\label{e:unierg}
\end{equation}

Denote by $\{\clF_n : n\ge 0\}$  the filtration associated with $\bfmX$.   The Markov property asserts that for bounded measurable functions $h\colon\state\to\Re$,
\[
\Expect[h(X_{n+1}) \mid \clF_n \,, X_n = x]  = \int P(x,dy) h(y)  
\]

In this paper, a stopping time $\tau : \Omega \to [0, \infty)$ is a random variable taking on values in the non-negative integers, with the defining property $\{\omega : \tau(\omega) \le n \,, \, \omega \in \Omega\} \in\clF_n$ for each $n \geq 0$.   A stationary policy is defined to be a measurable function $\phi\colon\state\to\{0,1\}$ that defines a stopping time:
\begin{equation}
\tau^\phi = \min \{n \ge 0 \colon \phi(X_n) = 1  \}
\label{e:tauphi}
\end{equation}

The optimal value function is defined as the infimum of \eqref{e:opt_stop} over all stopping times:  
for any $x\in\state$, 
\begin{equation}
h^*(x) \eqdef \inf_{\tau  } \Expect \big[\sum_{n = 0}^{\tau-1} \beta^n c(X_n) + \beta^{\tau} \cstop(X_{\tau}) | X_0 = x \big]   
\label{e:opt_stop_h_opt}
\end{equation}
Similarly, the associated \textit{Q-function} is defined as 
\[
Q^*(x) \eqdef c(x) + \beta  \Expect [h^*(X_{1}) \mid X_0 = x ]
\] 
It follows that $Q^*$ solves the associated Bellman equation \cite{tsiroy99}: for each $x \in \state$,
\begin{equation}
Q^*(x) \!=\! c(x) + \beta \Expect [ \min(\cstop(X_1),Q^*(X_1)) | X_0 = x ]
\label{e:QOSBell}
\end{equation}  
and  the optimal stopping rule is defined by the corresponding stationary policy,  
\begin{equation}
\phi^*(x) = \ind \{ \cstop(x) \le Q^*(x) \}   
\label{e:opt-policy}
\end{equation}
{where $\ind \{\cdot\}$ denotes the indicator function.} Using the general definition \eqref{e:tauphi}, an optimal stopping time satisfies $\tau^* = \tau^{\phi^*}$.



The Bellman equation \eqref{e:QOSBell} can be expressed as the functional fixed point equation: $Q^* = F Q^*$,
where $F$ denotes the \emph{dynamic programming operator}: for any function $Q: \state \to \Re$, and $x \in \state$,
\begin{equation}
\label{e:DP}
F Q(x)  \eqdef  c (x) + \beta \Expect [ \min(\cstop(X_1), Q(X_1)) | X_0 = x ]
\end{equation}


Analysis is framed in the usual Hilbert space $L_2(\pi)$ of real-valued measurable functions on $\state$ with inner product:
\begin{equation}
\langle f, g \rangle_\pi = \Expect [ f(X) g(X)] \,,
\label{e:pi_inner}
\end{equation}
and norm:
\begin{equation}
\| f \|_\pi = \sqrt{\langle f, f \rangle_\pi} \,,
\label{e:pi_norm}
\end{equation}
where the expectation in \eqref{e:pi_inner} is with respect to the steady state distribution $\pi$.
It is assumed throughout that the cost functions $c$ and $\cstop$ are in $L_2(\pi)$. 

\subsection{Objective}
\label{s:obj}

The goal in this work is to approximate $Q^*$ using a parameterized family of functions $\{Q^\theta\}$, where $\theta \in \Re^d$ denotes the parameter vector. We   restrict to linear parameterization throughout, so that:
\begin{equation}
Q^{\theta}(x) \eqdef \theta^\transpose \psi(x) , \qquad x \in \state
\label{e:Qtheta_linear}
\end{equation}
where $\psi \eqdef [\psi_1,\,\ldots, \psi_d]^\transpose$ with $\psi_i: \state \to \Re$ , $\psi_i \in L_2(\pi)$, $1 \leq i \leq d$, denotes the basis functions.   For any parameter vector $\theta \in \Re^d$, we denote the Bellman error  
\[
\clB_\clE^\theta =
F Q^{\theta} - Q^{\theta}.
\]

It is assumed that the basis functions are linearly independent: The $d \times d$ {covariance} matrix $\Sigma^\psi$ is full rank, where
\begin{equation}
\Sigma^\psi(i, j) = \langle \psi_i, \psi_j \rangle_{\pi}, \quad 1 \leq i,j \leq d
\label{e:Psi}
\end{equation}

In a finite state-space setting, it is possible to construct a consistent algorithm that computes the Q-function exactly \cite{devmey17a}. The Q-learning algorithm of Watkins \cite{wat89, watday92a} can be used in this case (see \cite{yuber07} for a discussion).

In a function approximation setting, we need to relax the goal of solving \eqref{e:QOSBell}. As in previous works \cite{tsiroy99,choroy06,yuber07}, the goal in this paper is to obtain the solution to a \emph{Galerkin relaxation} of \eqref{e:QOSBell}: Find $\theta^*$ such that,  
\begin{equation}
\Expect[ \clB_\clE^{\theta^*} \!\!(X_n)   \psi_i(X_n) ]
=0 \,, \quad 1 \leq i \leq d,
\label{e:eligQalgOS}
\end{equation}
or equivalently,
\begin{equation}
\langle F Q^{\theta^*} - Q^{\theta^*} , \psi_i \rangle_{\pi}  = 0 \,, \quad 1 \leq i \leq d.
\label{e:Q0algOS}
\end{equation}

In \cite{tsiroy99}, the authors show that the solution to the fixed point equations in \eqref{e:Q0algOS} satisfies (see \cite[Theorem 2]{tsiroy99}):
\[
\| Q^{\theta^*} - Q^* \|_\pi \leq \frac{1}{1 - \beta^2} \Big[  \min_\theta \| Q^\theta  -  Q^* \|_\pi  \Big].
\]

\ad{\rd{This is for finite state-action space --- OK to cite?}}

\subsection{Literature Survey}

Obtaining an approximate solution to the original problem \eqref{e:QOSBell} using a modified objective \eqref{e:Q0algOS} was first considered in \cite{tsiroy99}. The authors propose an extension of the TD($0$) algorithm of \cite{sut88, tsiroy97a}, and obtain convergence results under the assumption of a finite state-action space setting. 

Though it is not obvious at first sight, the algorithm in \cite{tsiroy99} is more closely connected to Watkins' Q-learning algorithm \cite{wat89,watday92a}, than the TD($0$) algorithm. This is specifically due to a minimum term that appears in \eqref{e:Q0algOS} (see definition of $F$ in \eqref{e:DP}), similar to what appears in Q-learning. This is important to note, because \emph{Q-learning algorithms are known to have convergence issues under function approximation settings}, and this is due to the fact that the dynamic programming operator may not be a contraction in general \cite{bertsi96a}. The operator $F$ defined in \eqref{e:DP} is quite special in this sense: it can be shown that it is a contraction with respect to the $\pi$-norm \cite{tsiroy99}:
\[
\| F Q - F Q' \|_\pi \leq \beta \| Q - Q' \|_\pi, \quad \text{for all } Q,Q' \in L_2(\pi)
\]

Since \cite{tsiroy99}, many other algorithms have been proposed to improve the convergence rate. In \cite{choroy06} the authors propose a matrix gain variant of the algorithm presented in \cite{tsiroy99}, improving the rate of convergence in numerical studies. In \cite{yuber07}, the authors take a ``least squares" approach to solve the problem, and propose the \emph{least squares Q-learning} algorithm, that has close resemblance to the least squares policy evaluation algorithm (LSPE ($0$) of \cite{nedber03a}). The authors recognize the high computational complexity of the algorithm, and propose alternative variants. In prior works \cite{choroy06} and \cite{yuber07}, though a function-approximation setting is considered, the state-space is assumed finite.

More recently, in \cite{devmey17a, devmey17b}, the authors propose the Zap Q-learning algorithm to solve for a solution to a fixed point equation similar to (but more general than) \eqref{e:QOSBell}. The proof of convergence is provided only for the finite state-action space setting, and more restrictively, a tabular basis is assumed (wherein the $\psi_i$'s span all possible functions).

\subsection{Contributions}
We make the following contributions in this work:
\begin{romannum}
	\item We extend the convergence analysis of Zap-Q-learning of \cite{devmey17b} to the problem of optimal-stopping \eqref{e:opt_stop} in a linear function approximation setting (the authors consider only a `tabular' basis in \cite{devmey17b}) 
	\item The algorithm and analysis presented in this work is superior to previous works on optimal stopping \cite{tsiroy99,choroy06,yuber07} in two ways: Firstly, the analysis in previous works only concern a finite state-action space setting; more importantly, the algorithm we propose has \emph{optimal asymptotic variance}, implying better convergence rates (see \Section{s:main_results} for a discussion and \Section{s:num} for numerical results).
\end{romannum}

The extension of the work \cite{devmey17b} to the current setting  is \emph{not trivial}: The tabular case is much simpler to analyze with lots of special structures, and in general, the theory for convergence of any Q-learning algorithm in a function approximation setting does not exist. Furthermore, the ODE analysis obtained in this paper (cf. Theorem~\ref{t:ZapOS}) provides great insights into the behavior of the Zap-Q algorithm, even in a linear function approximation setting.

The remainder of the paper is organized as follows:   \Section{s:alg} contains the approximation architecture, and introduces the Zap-Q-learning algorithm. The assumptions and main results are contained in \Section{s:main_results}. 
\Section{s:proof} provides a high-level proof of the results,  numerical results are collected together in
\Section{s:num}, and conclusions in \Section{s:conc}.  Full proofs are available in the extended version of this paper,  available on arXiv \cite{chedevbusmey19b}.


\section{Q-learning for Optimal Stopping}
\label{s:alg}

\subsection{Notation}

The following notation is useful for the convergence analysis. For each $\theta \in \Re^d$, we denote $\phi^{\theta}: \state \to \{0,1\}$ to be the corresponding policy:
\begin{equation}
\phi^{\theta} (x) \eqdef \ind{\{\cstop(x) \leq Q^{\theta}(x) \}}
\label{e:phitheta}
\end{equation}
For any function $f$ with domain $\state$, the operators $\clS_\theta$ and $\clS^c_\theta$ are defined  as the simple products,
\begin{subequations} 
	\begin{align}
	\clS_\theta f (x) & \eqdef  \ind{\{Q^{\theta}(x) < \cstop(x)\}} f (x)
	\label{e:clStheta} 
	\\
	\clS_\theta^c f (x) & \eqdef  \ind{\{\cstop(x) \leq Q^{\theta}(x)\}} f (x) 
	\label{e:clSthetac}
	\end{align}
\end{subequations}
Observe that for each $x \in \state$, $\clS_\theta f (x) = (1-\phi^{\theta} (x) )f (x) $.   

The objective  \eqref{e:Q0algOS} can then be expressed:  
\begin{equation}
A(\theta^*) \theta^* + \beta \barcstop(\theta^*)  + b^* = 0 \,,
\label{e:Athetabtheta}
\end{equation}
where, for each $\theta \in \Re^d$, $A(\theta)$ is a $d\times d$ matrix, and $b^*$ and $\barcstop(\theta)$ are $d$-dimensional vectors:  
\begin{align}
\!\!\!\!
A(\theta) &\eqdef \Expect[ \psi(X_n) \beta \clS_{\theta} \psi^\transpose (X_{n+1}) - \psi(X_n) \psi^\transpose(X_n) ]
\label{e:Atheta} 
\\
\!\!\!
b^* & \eqdef \Expect[ \psi(X_n) c(X_n) ] 
\label{e:bdef} 
\\
\!\!\!\!
\barcstop(\theta) & \eqdef \Expect[ \psi(X_n) \clS_\theta^c  \cstop(X_{n+1}) ] 
\label{e:cbdef} 
\end{align}


\subsection{Zap Q-Learning}

Before we introduce our main algorithm, it is useful to first consider a more general class of ``matrix gain" Q-learning algorithms.
Given a $d \times d$ matrix gain sequence $\{G_n: n \geq 0\}$, and a scalar step-size sequence $\{ \alpha_n: n \geq 0 \}$, the corresponding \emph{matrix gain Q-learning algorithm} for optimal stopping is given by the following recursion:
\begin{equation}
\begin{aligned}
\theta_{n+1}  & = \theta_n + \alpha_{n+1}G_{n+1} \psi(X_n) \td_{n+1}
\end{aligned}
\label{e:matrix_gain_Q}
\end{equation} 
with $\{\td_{n}\}$ denoting the ``temporal difference" sequence:
\[
\td_{n+\!1} \eqdef c(X_n)  +  \beta \! \min(\cstop(X_{n +\!1}),Q^{\theta_n}(X_{n + \! 1}))  - Q^{\theta_n}(X_n)
\]

The algorithm proposed in \cite{tsiroy99} is \eqref{e:matrix_gain_Q}, with $G_n \equiv I$ (the $d \times d$ identity matrix). This is similar to the TD($0$) algorithm \cite{tsiroy97a, sut88}.

\ad{Need citation here? For ``standard stochastic approximation \\ "? Or is it OK?}
The \emph{fixed point Kalman filter} algorithm of \cite{choroy06} can also be written as a special case of \eqref{e:matrix_gain_Q}: We have $G_n \equiv [\widehat{\Sigma}^\psi_n]^{\dagger}$, where $M^\dagger$ denotes the pseudo-inverse of any matrix $M$, and $\widehat{\Sigma}^\psi_n$ is an estimate of the mean $\Sigma_\psi$ defined in \eqref{e:Psi}; The estimate can be recursively obtained using standard Monte-Carlo recursion:
\begin{equation}
\widehat{\Sigma}^\psi_{n + 1} = \widehat{\Sigma}^\psi_n + \alpha_{n+1} \big[  \psi(X_n) \psi^\transpose (X_n)  - \widehat{\Sigma}^\psi_n \big]
\label{e:kalmangain}
\end{equation}

In the Zap-Q algorithm, the matrix gain sequence $\{G_n\}$ is designed so that the \emph{asymptotic covariance} of the resulting algorithm is minimized (see \Section{s:main_results} for details). It uses a matrix gain $G_n = -\haA_{n+1}^{\dagger}$, with $\haA_{n+1}$ being an estimate of $A(\theta_n)$, and $A(\theta)$ defined in \eqref{e:Atheta}.
\ad{Does pseudo inverse need projection?}

The term inside the expectation in \eqref{e:Atheta}, following the substitution $\theta = \theta_n$, is denoted
\begin{equation}
A_{n+1}   \eqdef  \psi(X_n) \bigl[\beta \clS_{\theta_n} \psi (X_{n+1}) - \psi(X_n)  \bigr]^\transpose
\label{e:An}
\end{equation}
Using \eqref{e:An}, the matrix $A(\theta_n)$ is recursively estimated using stochastic approximation in the Zap-Q algorithm:
\begin{algorithm}[H]
	\caption{\em Zap-Q for Optimal Stopping}
	\label{a:ZapOptStop}
	\begin{algorithmic}[1]
		\INPUT Initial $\theta_0  \in\Re^d$, $\haA_0$:    $d\times d$, negative definite;  step-size sequences $\{\alpha_n\}$ and $\{\gamma_n\}$
		and $n=0$
		\Repeat
		\State 
		Obtain the \emph{Temporal Difference} term:
		{\small{
				\[
				\hspace{-0.1in}\td_{n+1} \!=\!  c(X_n)   + \beta  \min(\cstop(X_{n+1}),Q^{\theta_n}(X_{n+1}))     - Q^{\theta_n}(X_n)
				\] 
			}
		}
		\State
		Update the matrix gain estimate $\haA_{n}$ of $A(\theta_n)$, with $A_{n+1}$ defined in \eqref{e:An}:
		\begin{equation}
		\haA_{n+1} = \haA_n + \gamma_{n+1}  \bigl[  A_{n+1} - \haA_n  \bigr]
		\label{e:QSNR2OSAdef}
		\end{equation}
		\State
		Update the parameter vector:
		\begin{equation}
		\theta_{n+1}  = \theta_n -\alpha_{n+1}\haA_{n+1}^{\dagger}  \psi(X_n) \td_{n+1}
		\label{e:QSNR2OSdef}
		\end{equation}
		\State $n = n+1$
		\Until{$n \geq N$}
		\OUTPUT $\theta = \theta_N$
	\end{algorithmic}
\end{algorithm} 


\subsection{Discussion}

Algorithm~\ref{a:ZapOptStop} belongs to a general class of algorithms known as \emph{two-time-scale stochastic approximation} \cite{bor08a}: the recursion in \eqref{e:QSNR2OSdef} on the \emph{slower} time-scale intends to estimate the parameter vector $\theta^*$, and for each $n \geq 0$, the recursion \eqref{e:QSNR2OSAdef} on the \emph{faster} time-scale intends to estimate the mean $A(\theta_n)$. The step-size sequences $\{\alpha_n\}$ and $\{\gamma_n\}$ have to satisfy the standard requirements for separation of time-scales \cite{bor08a}: for any $\varrho \in (0.5,1)$, we choose
\begin{equation}
\label{e:GAINSos}
\alpha_n = {1}/ {n} \,, \qquad \gamma_n =  {1}/{n^\varrho} 
\end{equation}


%

For each $\theta\in\Re^d$, consider the following terms: 
\begin{subequations}
	\begin{align}
	b(\theta)& = - A(\theta) \theta - \beta \barcstop(\theta)
	\label{e:btheta}
	\\
	c^\theta(x) & =  Q^{\theta}(x)
	\label{e:ctheta}
	\\
	& \hspace{0.1in} -\Expect \big [\beta \min(c_s(X_{n + 1}), Q^{\theta}(X_{n + 1}))  \! \mid \! X_{n} = x \big]
	\nonumber
	\end{align}
	\label{e:b+ctheta}
\end{subequations}%
The vector $b(\theta)$ is analogous to $b^*$ in   \eqref{e:Athetabtheta},  and \eqref{e:ctheta} recalls the Bellman equation 
\eqref{e:QOSBell}. The following \Prop{t:projCost} is direct from these definitions.  It shows that  $b(\theta)$ is the ``projection'' of the cost function $c^\theta$, similar to how $b^*$ is related to $c$ through \eqref{e:bdef}.

\begin{proposition}
	\label{t:projCost}
	For each $\theta \in \Re^d$, we have:
	\begin{equation}
	b(\theta) = \Expect \big[ c^\theta(X_n) \psi(X_n) \big]
	\end{equation}
	where the expectation is in steady state.  In particular, 
	\[ b^*   = b(\theta^*) \]
\end{proposition}
\qed




\section{Assumptions and Main Results}
\label{s:main_results}

\subsection{Preliminaries}

We first summarize preliminary results here that will be used to establish the main results in the following sections. The proofs of all the technical results are contained in the Appendix of \cite{chedevbusmey19b}. 

We start with the contraction property of the operator $F$ defined in \eqref{e:DP}. The following is a result directly obtained from \cite{tsiroy99} (see \cite[Lemma $4$ on p. $1844$]{tsiroy99}).
\begin{lemma}
	\label{t:DPcont}
	The dynamic programming operator $F$ defined in \eqref{e:DP} satisfies:
	\[
	\| F Q - F Q' \| \leq \beta \| Q - Q' \|, \qquad Q,Q' \in L_2(\pi).
	\]
	Furthermore, $Q^*$ is the unique fixed point of $F$ in $L_2(\pi)$. 
	\qed
\end{lemma}

Recall that $Q^\theta: \state \to \Re$ is defined in \eqref{e:Qtheta_linear}. Similar to the operator $F$, for each $\theta \in \Re^d$ we define operators $H^\theta$ and $F^\theta$ that operate on functions $Q: \state \to \Re$ as follows:
\begin{eqnarray}
H^\theta Q (x) & = &
\begin{cases}
Q(x), \qquad \text{if} \quad Q^\theta (x) < \cstop(x)  \\
\cstop(x), \qquad \text{otherwise}
\end{cases}
\\
F^\theta Q & = &  c + \beta P H^\theta Q.
\label{e:Ftheta}
\end{eqnarray}

The following Lemma is a slight extension of \Lemma{t:DPcont}. 
\vspace{0.1in}
\begin{lemma}
	\label{t:Fthetacont}
	For each $\theta \in \Re^d$, the operator $F^\theta$ satisfies:
	\[
	\| F^\theta Q - F^\theta Q' \| \leq \beta \| Q - Q' \|, \qquad \,\, Q,Q' \in L_2(\pi)
	\]
	\qed
\end{lemma}

The next result is a direct consequence of \Lemma{t:Fthetacont}, and establishes the inveritbility of the matrix $A(\theta)$ for any $\theta \in \Re^d$:
\begin{lemma}
	\label{t:AthetaNegEig}
	For each $\theta \in \Re^d$,
	\begin{enumerate}[(i)]
		\item The $d \times d$ matrix $A(\theta)$ defined in \eqref{e:Atheta} satisfies:
		\begin{equation}
		-
		v^\transpose A(\theta) v \ge  (1-\beta) v^\transpose \Sigma_\psi v, 
		\label{e:AthetaNegEig}
		\end{equation}
		for each $v \in \Re^d$, with $\Sigma_\psi$ defined in \eqref{e:Psi}.
		\item Eigenvalues of $A(\theta)$ are strictly bounded away from 0, and $\{A^{-1}(\theta) \! : \! \theta \! \in \! \Re ^d\}$ is uniformly bounded.
	\end{enumerate}
	\qed
\end{lemma}


\Prop{t:projCost} implies   a Lipschitz bound on the function  $b$ defined in \eqref{e:btheta}:
\begin{lemma}
	\label{t:bthetaLip}
	The mapping $b$ is Lipschitz:
	For some $\ell_1 > 0$, and each $\theta^1, \theta^2 \in \Re^d$,
	\[
	\| b(\theta^1) - b(\theta^2) \|  \leq \ell_1 \|\theta^1 - \theta^2 \|
	\]
	\qed
\end{lemma}

\subsection{Assumptions \& Main Result}

The following assumptions are made throughout:

\assume{Assumption A1:}  $\bfmX$ is a uniformly ergodic Markov chain on the compact state space $\state$, with a unique invariant probability measure, $\pi$ (cf. \eqref{e:unierg}).  

\assume{Assumption A2:}  There exists a unique solution $\theta^*$ to the objective \eqref{e:Q0algOS}. 

\assume{Assumption A3.1:}  
The conditional distribution of $\psi(X_{n+1})$ given $X_n = x$ has a density, $f_{\psi \mid x}(z)$. This density is also assumed to have uniformly bounded likelihood ratio $\ell(z \mid x )$ with respect to the Gaussian density $\mathcal N(\psi(x), I)$.  

\assume{Assumption A3.2:}  
It is assumed moreover that the  function $c_s-1$ is   in the span of $\{\psi_i\}$.


\assume{Assumption A4:}
The parameter sequence 
$\{\theta_n:n\geq 1\}$ is bounded {\em a.s.}.

Assumption {\textbf{A3}} consists of technical conditions required for the proof of convergence. The density assumption is imposed to ensure that the conditional expectation given $X_n$ of functions such as $\clS_{\theta} \psi^\transpose (X_{n+1}) $ are smooth as a function of $\theta$.  Furthermore, it implies that   $\Sigma_\psi$ is positive definite.

Assumption {\textbf{A4}} is a standard assumption in much of the recent stochastic approximation literature.
We conjecture that the boundedness can be established via  an extension of the results in \cite{bormey00a, bor08a}. The ``ODE at infinity'' posed there is stable as required, but the extension of the results to the current setting of two time-scale stochastic approximation with Markovian noise is the only challenge.  
\spm{did we look?  I bet he has an extension!}
\ad{\rd{Sep. 2019: We skimmed through, and could not find anything relevant. They either assume boundedness to prove convergence, or they assume MDS noise. A4 seems to be a common assumption in Bhatnagar's works.}}


The main result of this paper establishes convergence of iterates $\{\theta_n\}$ obtained using Algorithm~\ref{a:ZapOptStop}:
\begin{theorem}
	\label{t:ZapOS}
	Suppose that Assumptions A1-A4 hold. Then,
	\begin{romannum}
		\item The parameter sequence $\{\theta_n\}$ obtained using the Zap-Q algorithm converges to $\theta^*$ a.s., where $\theta^*$ satisfies \eqref{e:Q0algOS}.
		\item An ODE approximation holds for the sequences $\{\theta_n,b(\theta_n)\}$ by continuous time functions $(\bfmw,\bfmb)$ satisfying
		\begin{equation}
		\begin{aligned}
		\ddt b(t) & = - b(t) + b \\
		b(t) & = - A(w(t)) w(t) - \beta \barcstop(w(t))
		\end{aligned}
		\label{e:ZapODEOS}
		\end{equation}
		\qed
	\end{romannum}
\end{theorem}

The term \textit{ODE approximation} is standard  in the SA literature:
For $t \geq s$, let $w^s(t)$ denote the solution to: 
\begin{equation}
\ddt w^s(t) = \xi(w^s(t)), \quad w^s(s) = \barw(s)
\label{e:ws}
\end{equation}
for some $\xi \colon \Re^d \to \Re^d$, and  $\barw(t)$ denoting the continuous time process constructed from the sequence  $\{\theta_n:n\geq 0\}$ via linear interpolation.
We say that the ODE approximation
$
\ddt w(t) = \xi (w(t))
$
holds for the sequence $\{\theta_n:n\geq 1\}$, if the following is true for any $T > 0$:
\begin{equation*}
\lim_{s\rightarrow\infty}\sup_{t\in[s, s+T]}\|\barw (t) -w^s(t)\| =0,\: a.s.
\label{e:barw_ODE}
\end{equation*}
Details are made precise in \Section{s:ODEOS}.
The optimality of the algorithm in terms of the asymptotic variance is discussed next.


\subsection{Asymptotic Variance}
\label{s:asym_var}


\def\SigmaTheta{\Sigma_{\Theta}}


The asymptotic covariance $\SigmaTheta$ of any algorithm is defined to be the following limit:
\begin{equation}
\SigmaTheta \eqdef \lim_{n\to\infty } n \Expect \big [ \big (\theta_n - \theta^* \big ) \big (\theta_n - \theta^* \big )^{\transpose} \big ]
\label{e:Sigma_theta}
\end{equation}

Consider the matrix gain Q-learning algorithm \eqref{e:matrix_gain_Q}, and suppose the matrix sequence $\{G_n: n \geq 0\}$ is constant:  $G_n\equiv G$.  
Also, suppose that all eigenvalues of $GA(\theta^*)$ satisfy  $\lambda \big( G A(\theta^*) \big) < - \half$. Following standard analysis (see Section 2.2 of \cite{devmey17a} and references therein), it can be shown that, under general assumptions, the asymptotic covariance of the algorithm \eqref{e:matrix_gain_Q} can be obtained as a solution to the Lyapunov equation:
\begin{equation}
\begin{aligned}
\big( G A(\theta^*) + \half I \big) \SigmaTheta & + \SigmaTheta \big( G A(\theta^*) + \half I \big)^\transpose 
\\
&\hspace{0.6in} +G \Sigma_\clE G^\transpose = 0
\end{aligned}
\label{e:gain_lyap}
\end{equation}
where $\Sigma_\clE$ is the ``noise covariance matrix", that is defined as follows.

A ``noise sequence'' $\{\clE_n\}$ is defined as 
\begin{equation}
\clE_{n} \eqdef \tilA_{n+1} \theta^* + \tilb_{n+1} + \tilA_{n+1} \tiltheta_n
\label{e:Deltan}
\end{equation}
where $\tilA_{n+1} \eqdef A_{n+1} - A(\theta^*)$, $\tilb_{n+1} \eqdef b_{n+1} - b^*$, \\ $\tiltheta_n \eqdef \theta_n - \theta^*$, with $A_{n+1}$ defined in \eqref{e:An}, $A(\theta)$ defined in \eqref{e:Atheta},
\begin{equation}
b_{n+1} \eqdef  \psi(X_n) \big[c(X_n) +  \clS^c_{\theta_n} \cstop(X_{n+1}) \big]
\label{e:bn}
\end{equation}
and $b^*$ defined in \eqref{e:bdef}.   
The noise covariance matrix is then defined as the limit\begin{equation}
\Sigma_\clE = \lim_{T \to \infty} \frac{1}{T} \Expect [S_T S_T^{\transpose}] 
\label{e:Sigma_Delta}
\end{equation}
in which $S_T = \sum_{n = 1}^{T} \clE_n$, and the expectation is in steady state.

\subsubsection*{\textbf{Optimality of the asymptotic covariance}}

The asymptotic covariance $\SigmaTheta$ can be obtained as a solution to \eqref{e:gain_lyap} only when all eigenvalues satisfy  $\lambda ( G A(\theta^*) ) < - \half$. If there exists at least one eigenvalue such that $\lambda ( G A(\theta^*) ) \geq - \half$, then, under general conditions, it can be shown that the asymptotic covariance is not finite \cite{devmey17a}.
This implies that the rate of convergence of $\theta_n$ is  \emph{slower} than $O(1/\sqrt{n})$. 

It is possible to optimize the covariance $\SigmaTheta$ over all matrix gains $G$ using \eqref{e:gain_lyap}. Specifically, it can be shown that letting $G^* = -A(\theta^*)$ will result in the \emph{minimum} asymptotic covariance $\Sigma^*$, where
\begin{equation}
\Sigma^* = A(\theta^*)^{-1} \Sigma_\clE \big( A(\theta^*)^{-1} \big)^\transpose
\label{e:SigmaOpt}
\end{equation}
That is, for any other gain $G_n\equiv G$, denoting $\SigmaTheta^G$ to be the asymptotic covariance of the algorithm \eqref{e:matrix_gain_Q} obtained as a solution to the Lyapunov equation \eqref{e:gain_lyap}, the difference $\SigmaTheta^G \! - \! \Sigma^*$ is positive semi-definite. This is specifically true for the algorithms proposed in \cite{tsiroy99} and \cite{choroy06}.

The Zap~Q algorithm is specifically designed to achieve the optimal asymptotic covariance.   A full proof of optimality will require extra effort.  
\Theorem{t:ZapOS} tells us that we have the required convergence  $G_n\to G^*$ for this algorithm.    
Provided we can obtain additional tightness bounds for the scaled error $\sqrt{n}\tiltheta_n$, we obtain a functional Central Limit Theorem with optimal covariance $\Sigma^*$ \cite{bor08a}.    Minor additional bounds ensure convergence of \eqref{e:Sigma_theta} to the optimal covariance $\Sigma^*$.

%
%
%
%
%
The next section is dedicated to the proof of \Theorem{t:ZapOS}.

\section{Proof of Theorem~\ref{t:ZapOS}}
\label{s:proof}


\subsection{Overview of the Proof}
\label{s:overview}

Unlike martingale difference assumptions in standard stochastic approximation, the noise in our algorithm is Markovian. The first part of this section establishes that our noise sequence satisfies the so called \emph{ODE friendly property} {\cite{tadic2003asymptotic}}:
A vector-valued sequence of random variables $\{\clE_k\}$ will be called \textit{ODE-friendly} if it admits the decomposition,
\begin{equation}
\clE_k   =  \Delta_k  +  \clT_k-\clT_{k-1} + \epsy_k\,,
\quad k\ge 1
\label{e:bor08aNoiseOS}
\end{equation}
in which:
\begin{enumerate}[(i)]
	\item $\{ \Delta_k : k\ge 1\}$ is a martingale-difference sequence satisfying  $ \Expect[\|  \Delta_{k} \|^2 \! \mid \! \clF_k] \le \bar\sigma^2_\Delta < \infty$ {\em a.s.}\ for all~$k$
	\item $\{ \clT_k : k\ge 1\}$ is a bounded sequence
	\item The final sequence $\{\epsy_k\}$ is bounded and satisfies:
	\begin{equation}
	\sum_{k=1}^\infty \gamma_k \|\epsy_k\| <\infty \quad a.s.\, .
	\label{e:Friendly3OS}
	\end{equation}
\end{enumerate}
Intuitively, if an error sequence satisfies the above properties, it can be shown that its asymptotic effect on the parameter update is zero. This allows us to argue that the matrix gain estimate $\haA_{n+1}$ is close to the mean $A(\theta_n)$ in the fast time-scale recursion \eqref{e:QSNR2OSAdef}. We then consider the slow time-scale recursion \eqref{e:QSNR2OSdef}, and obtain the ODE approximations for $\{\theta_n\}$ and the expected projected cost $\{ b(\theta_n) \}$. The fact that these ODE's are stable (with a unique stationary point) will then establish the convergence of the algorithm.

\subsection{ODE Analysis}
\label{s:ODEOS}

The remainder of this section is dedicated to the proof of the ODE approximation \eqref{e:ZapODEOS}.
The construction of an approximating ODE involves first defining a continuous time process $\overline{\bfmw}$.  Denote
\begin{equation}
t_n = \sum_{i=1}^n \alpha_i,\ \  n\ge 1, \quad t_0=0\,,
\label{e:tnOS}
\end{equation}
and define $\barw({t_n}) = \theta_n$ at these time points, with the definition extended to $\Re_+$ via linear interpolation.

Along with the piecewise linear continuous-time process   $\{ \barw_t : t\ge 0 \}$,   denote by $\{\barclA_t : t\ge 0\}$   the  piecewise linear continuous-time process defined similarly, with $\barclA_{t_n} = \haA_n$,  $n\ge 1$.
Furthermore, for each $t\ge 0$, denote  
\[
\barb_t \equiv b(\barw_t) \eqdef- A(\barw_t) \barw_t - \beta \barcstop(\barw_t)
\]


To construct an ODE, it is convenient first to obtain an alternative and suggestive representation for the pair of equations (\ref{e:QSNR2OSAdef},\ref{e:QSNR2OSdef}).   



\Lemma{t:pre-ODEOS} establishes that the error sequences that appear in the updates for $\{\theta_n\}$ and $ \{ \haA_n \}$ are ``ODE friendly".
\begin{lemma}
	\label{t:pre-ODEOS}
	The pair of equations (\ref{e:QSNR2OSAdef}, \ref{e:QSNR2OSdef}) can be expressed,
	\begin{subequations}
		\begin{align}
		\theta_{n+1} & \!=\! \theta_n \!-\! \alpha_{n+1} \haA^{\dagger}_{n+1} \bigl[ A(\theta_n) \theta_n \! + \! \beta  \barcstop(\theta_n) \! + \! b^*
		\\
		\hspace{-0.2in}&\qquad \qquad  \qquad \qquad +\clE^A_{n+1} \theta_n + \clE^\theta_{n+1} \bigr] 
		\nonumber
		\\[.1cm]
		\haA_{n+1} & = \haA_n + \gamma_{n+1}  \bigl[ A(\theta_n)  - \haA_n +\clE^A_{n+1}   \bigr]
		\label{e:SNR2OSlinearSA_preODE_A}
		\end{align}
		\label{e:SNR2OSlinearSA_preODE}%
	\end{subequations}%
	in which the sequences $\{ \clE^\theta_n\,, \clE^A_n  : n\ge 1\}$  are ODE-friendly.
	\qed
\end{lemma}

%

The following result establishes that $\haA_n$ recursively obtained by \eqref{e:QSNR2OSAdef} approximates the mean $A(\theta_n)$: 
\begin{lemma}
	Suppose the sequence $\{\mathcal E^A_{n}:n \geq 1\}$ is ODE-friendly. Then, 
	\begin{enumerate}[(i)]
		\item $\displaystyle \lim_{n\rightarrow\infty}\|\haA_n - A(\theta_n) \| = 0,\; a.s.$
		\item Consequently, $\haA_n^{\dagger} \neq \haA_n^{-1}$ only finitely often, and $\displaystyle \lim_{n\rightarrow\infty}\|\haA_n^\dagger - A^{-1}(\theta_n) \| = 0,\; a.s.$
	\end{enumerate}
	\label{t:odegain}
	\qed
\end{lemma}



With the definition of ODE approximation below \eqref{e:ws}, we have:
\begin{lemma}
	\label{t:odetheta}
	The ODE approximation for $\{\theta_n \}$  holds:
	with probability one, the piece-wise continuous function $\barw(t)$ asymptotically tracks the  ODE:
	\begin{equation}
	\begin{aligned}
	\ddt {w}(t) = -A^{-1}(w(t)) \big [ b^*-b(w(t)) \big] 
	\end{aligned}
	\label{eq:thetaODEsubseq}
	\end{equation}
	\qed
\end{lemma}


For a fixed (but arbitrary) time horizon $T>0$, we define two families of uniformly bounded and uniformly Lipschitz continuous functions: $\{\barw(s+t),t\in[0, T]\}_{s\geq 0}$ and $\{\barb(s+t),t\in[0, T]\}_{s\geq 0}$. Sub-sequential limits of $\{\barw(s+t),\,t\in[0, T]\}_{s\geq 0}$ and $\{\barb(s+t),t\in[0, T]\}_{s\geq 0}$ are denoted  $w_t$ and $b_t$ respectively.

We recast the ODE limit of the projected  cost as follows:
{
	\begin{lemma}
		\label{t:b_ode}
		For any sub-sequential limits $\{w_t, b_t\}$,
		\begin{enumerate}[(i)]
			\item they satisfy $  b_t = b(w_t)$.
			\item for {\em a.e.} $t\in[0,T]$,
			\begin{equation}
			\begin{aligned}
			\frac{d}{dt}b_t 
			& = -A(w_t)\frac{d}{dt}w_t
			= - b_t + b^*
			\end{aligned}
			\label{eq:costODE}
			\end{equation}
		\end{enumerate}
		\qed
	\end{lemma}
}

\paragraph*{Proof of \Theorem{t:ZapOS}}
Bounedness of sequences $\{\haA_{n}:n\geq 0\}$ and $\{\haA_n^{-1}:n\geq 0\}$ is established in Lemma \ref{t:odegain}. Together with boundedness assumption of $\{\theta_n: n\geq 0\}$, the ODE approximation is established in Lemma \ref{t:b_ode}. Result (i) then follows from those two results using standard arguments from \cite{bor08a}.
\hfill$\blacksquare$


\section{Numerical Results}
\label{s:num}

\begin{figure*}[h]
	\Ebox{1}{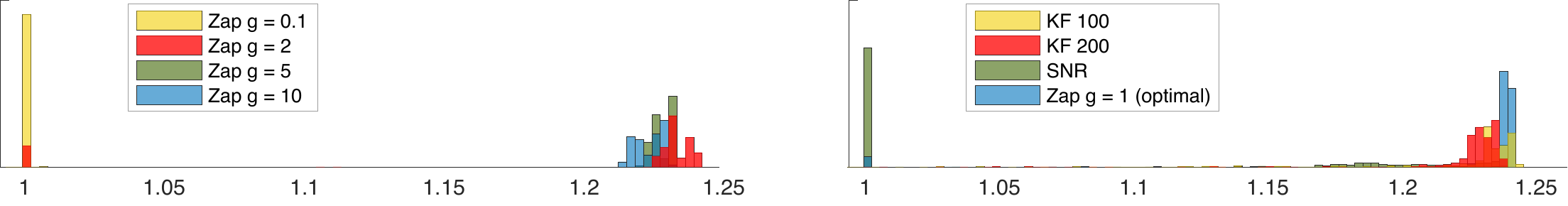}
	\caption{\small  Average rewards obtained using various Q-learning algorithms}
	\label{AvgReward_2e6_All_combined}
\end{figure*}

In this section we illustrate the performance of the Zap Q-learning algorithm in comparison with existing techniques, on a finance problem that has been studied in prior work~\cite{choroy06, tsiroy99}. We observe that the Zap algorithm performs very well, despite the fact that some of the technical assumptions made in \Section{s:main_results} do not hold. 

\subsection{Finance model}

The following finance example is used in \cite{choroy06, tsiroy99} to evaluate the performance of their algorithms for optimal stopping. The reader is referred to these references for complete details of the problem set-up. 

The Markovian state process considered is the vector of ratios: 
$
X_n = (\tilp_{n-99},\tilp_{n-98}, \ldots,\tilp_{n} )^\transpose / \tilp_{n-100}$, $n \geq 0
$,
in which $\{\tilp_t : t\in\Re\}$ is a geometric Brownian motion (derived from an exogenous price-process).
This uncontrolled Markov chain is positive Harris recurrent on the state space $\state \equiv \Re^{100}$, so $\state$ is not compact. The Markov chain is uniformly ergodic.

The ``time to exercise'' is modeled as a stopping time $\tau \in \bfmZ^+$.   The associated expected reward is defined as $
\Expect [ \beta^{\tau} r(X_\tau) ]$,  with $r(X_n) \eqdef X_n(100) = \tilp_{n} / \tilp_{n-100}$ and $\beta \in (0,1)$ fixed.   The objective of finding a policy that maximizes the expected reward is modeled as an optimal stopping time problem.

The value function is defined to be the infimum \eqref{e:opt_stop_h_opt}, with $c \equiv 0$ and $\cstop \equiv - r$ 
{(the objective in \Section{s:intro} is to minimize the expected cost, while here, the objective is to maximize the expected reward)}.
{
	The associated Q-function is defined using \eqref{e:QOSBell}, and the associated optimal policy using \eqref{e:opt-policy}:
	$
	\phi^*(x) = \ind \{ r(x) \ge Q^*(x) \}.
	$}
\archive{I am so suspicious of this claim. We must discuss}

When the Q-function is linearly approximated using \eqref{e:Qtheta_linear}, for a fixed parameter vector $\theta$, the associated value function can be expressed: 
\begin{equation}
\begin{aligned} 
h_{\phi^\theta}(x) & \eqdef \Expect [\beta^{\tau_\theta } r(X_{\tau_\theta }) \mid x_0 = x ]\,,
\end{aligned}
\label{e:Jtheta}
\end{equation}
where,
\begin{equation}
\begin{aligned} 
\tau_\theta & \eqdef \min \{n \colon  \phi^\theta (X_n) = 1 \}
\\
\phi^\theta(x) & \eqdef \ind \{ r(x)  \ge Q^\theta(x) \}
\end{aligned}
\label{e:phi_theta}
\end{equation}
Given a parameter estimate $\theta$ and the initial state $X(0) = x$, the corresponding average reward $h_{\phi^\theta}(x)$  was estimated using Monte-Carlo in the numerical experiments that follow. 

\subsection{Approximation \& Algorithms}

Along with Zap Q-learning algorithm we also implement the \emph{fixed point Kalman filter} algorithm of \cite{choroy06} to estimate $\theta^*$. This algorithm is given by the update equations \eqref{e:matrix_gain_Q} and \eqref{e:kalmangain}. 
{The computational as well as storage complexities of the least squares Q-learning algorithm (and its variants) \cite{yuber07} are too high for implementation.}

\subsection{Implementation Details}

The experimental setting of \cite{choroy06, tsiroy99} is used to define the set of basis functions and other parameters. 
\archive{\today:  We will need to be sure there is no ambiguity about which parameters were used.   Did  
	tsiroy99,choroy06 settle on these values?
	\\
	were used, with volatility factor $\sigma = 0.02$ and the short term interest rate $\interest = 0.0004$. }
We choose the dimension of the parameter vector $d = 10$, with the basis functions defined in \cite{choroy06}. 
The objective here is to compare the performances of the fixed point Kalman filter algorithm with the Zap-Q learning algorithm in terms of the resulting average reward \eqref{e:Jtheta}. 

Recall that the step-size for the Zap Q-learning algorithm is given in \eqref{e:GAINSos}. We set $
\gamma_n = {n^{-0.85}}
$
for all implementations of the Zap algorithm, but similar to what is done in \cite{choroy06}, we experiment with different choices for $\alpha_n$. 
Specifically, in addition to $\alpha_n = n^{-1}$, we let:
\begin{equation}
\alpha_n = \frac{g}{b + n}
\label{e:bvr_ss}
\end{equation}
with $b = 1e4$ and experiment with $g = 2, 5,$ and $10$. In addition, we also implement Zap with $(\alpha_n = 0.1/n \,, \gamma_n = 1/n^{0.85})$. Based on the discussion in \Section{s:asym_var}, we expect this choice of step-size sequences to result in infinite asymptotic variance. 

In the implementation of the {fixed point Kalman filter} algorithm, as suggested by the authors, we choose step-size $\alpha_n\!=\!1/n$ for the matrix gain update rule in \eqref{e:kalmangain}, and step-size of the form \eqref{e:bvr_ss} for the parameter update in \eqref{e:matrix_gain_Q}. Specifically, we let $b = 1e4$, and $g = 100$ and~$200$. 

The number of iterations for each of the algorithm is fixed to be $N = 2e6$.

\subsection{Experimental Results}

\begin{figure*}
	\Ebox{1}{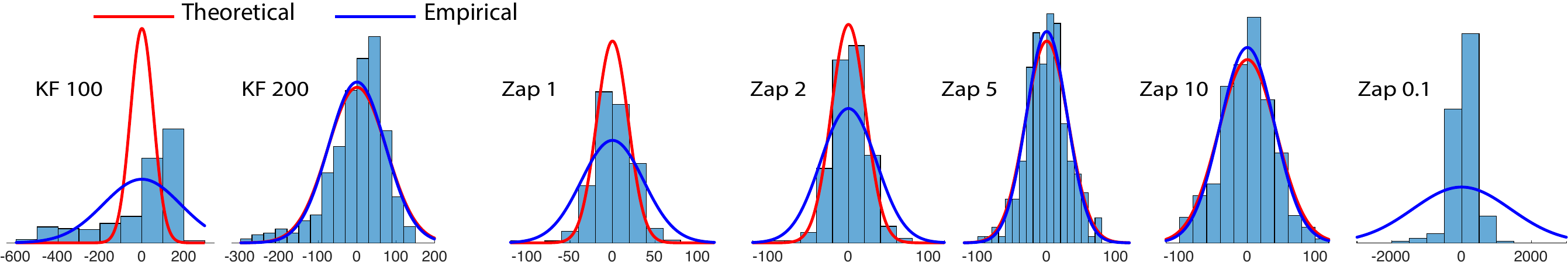}
	\caption{\small   Asymptotic variance for various Q-learning algorithms for optimal stopping in finance}
	\label{asym_var_plot}
\end{figure*}

The average reward histogram was obtained by the following steps: We simulate $500$ parallel simulations of each of the algorithms to obtain {as many} estimates of $\theta^*$. Each of these estimates defines a policy $\phi^{\theta_N}$ defined in \eqref{e:phi_theta}. We then estimate the corresponding average reward $h_{\phi^{\theta_N}}(x)$ defined in \eqref{e:Jtheta}, with $X(0) = x = 1$. 

Along with the average discounted rewards, we also plot the histograms to visualize the asymptotic variance \eqref{e:Sigma_theta}, for each of the algorithms. The theoretical values of the covariance matrices $\Sigma^*$ and $\Sigma^G_\Theta$ were estimated through the following steps:
The matrices $A(\theta^*)$ and $\Sigma_\psi$ (the limit of the matrix gain used in \cite{choroy06}) were estimated via Monte-Carlo. Estimation of $A(\theta^*)$ requires  an estimate of $\theta^*$; this was taken to be  $\theta_N$ obtained using the Zap-Q algorithm with $\alpha_n = n^{-1}$ and $\gamma_n = n^{-0.85}$.  This estimate of $\theta^*$  was also used to estimate the  covariance matrix $\Sigma_\Delta$ defined in \eqref{e:Sigma_Delta} using the batch means method.  
The matrices $\Sigma^G_\Theta$ and $\Sigma^*$    were then obtained using \eqref{e:gain_lyap} and  \eqref{e:SigmaOpt}, respectively.

\Fig{AvgReward_2e6_All_combined}  contains the histograms of the average rewards obtained using the above algorithms. \Fig{asym_var_plot} contains the histograms of $\sqrt{N} \big(\theta_N(8) - \theta^*(8) \big)$ 
{along with a plot of the theoretical prediction.} 


It was observed that the eigenvalues of the matrix $A(\theta^*)$ have a wide spread: The condition-number is of the order $10^4$. 
Despite a badly conditioned matrix gain, it is observed in Fig.~\ref{AvgReward_2e6_All_combined}, that the average rewards of the Zap-Q algorithms are better than its competitors. It is also observed that the algorithm is robust to the choice of step-sizes. In \Fig{asym_var_plot} we observe that the asymptotic behavior of the algorithms is close match to the theoretical prediction. Specifically,  large variance of Zap-Q with step-size $\alpha_n = 0.1 /n$ confirms that the asymptotic variance is very large (ideally, infinity), if the eigenvalues of the matrix $G A(\theta^*) > -\half$.


\section{Conclusion} 
\label{s:conc}

In this paper, we extend the theory for the Zap Q-learning algorithm to a linear function approximation setting, with application to optimal stopping. We prove convergence of the algorithm using ODE analysis, and also observe that it achieves optimal asymptotic variance.  The extension of the previous analysis to the current setting  is \emph{not trivial}: Analysis of Zap-Q in the tabular case is much simpler, with lots of special structures, and in general, the theory for convergence of Q-learning algorithms in a function approximation setting does not exist. More importantly, we believe that the ODE analysis obtained in this paper provides important insights into the behavior of the Zap-Q algorithm, even in a function approximation setting. This may be a starting point for analysis of Q-learning algorithms in general function-approximation settings, which is an on-going work.



\bibliographystyle{IEEEtran}


\newpage
\onecolumn
\section{APPENDIX}
\setcounter{section}{7}

\paragraph*{Proof of \Lemma{t:Fthetacont}}
Based on the definition \eqref{e:Ftheta}, we have:
\[
\begin{aligned}
\| F^\theta Q - F^\theta Q' \| 
& = \beta \|P H^\theta Q  -  P H^\theta Q'\|
\\
& \leq \beta \|H^\theta Q  -  H^\theta Q'\|
\\
& \leq \beta \|Q - Q' \| \,,
\end{aligned}
\]
where the first inequality follows from the fact that $\|P\| \leq 1$ (with $\| P\|$ being the induced operator norm in $L_2(\pi)$). The last inequality is true because:
\[
H^\theta Q (x) - H^\theta Q'(x) =
S_\theta \big (Q - Q' \big)(x)\,, \quad x \in \state 
\]
\qed

\paragraph*{Proof of \Lemma{t:AthetaNegEig}}
To show result (i), we rewrite $A(\theta)$ as the difference of two matrices, $A(\theta) = A_{NL}(\theta) - A_L$, denoting $A_{NL}(\theta)$ to be the part of the matrix that depends on $\theta$ and $A_L$ to be the one that is independent of $\theta$:
\[
\begin{aligned}
A_{NL}(\theta) & \eqdef \Expect[ \psi(X_n) \beta S_\theta \psi^\transpose (X_{n+1})] \\ 
A_L  & \eqdef \Expect[\psi(X_n) \psi^\transpose(X_n)]
\end{aligned}
\]
Proving \eqref{e:AthetaNegEig} is equivalent to proving:
\[
v^\transpose A_{NL}(\theta) v  - v^\transpose A_L v  \leq (\beta - 1) v^\transpose \Sigma_\psi v, \quad v\in \Re ^d.
\]
The proof is easier to follow if we suppose that the vector $v$ is a difference of two parameter vectors, $v = \theta^1 - \theta^2$. 
Expanding the left hand side of the above inequality:
\[
\begin{aligned}
v^\transpose A_{NL}(\theta) v 
& = (\theta^1 - \theta^2)^\transpose \Expect \big[ \psi(X_n) \beta P H^\theta \big( Q^{\theta^1}(X_n) - Q^{\theta^2} (X_n) \big) \big ]
\\
& = \beta \Expect \big[ \big( Q^{\theta^1}(X_n) - Q^{\theta^2} (X_n) \big) P H^\theta \big( Q^{\theta^1}(X_n) - Q^{\theta^2} (X_n) \big) \big ].
\end{aligned}
\]
Next, using Cauchy-Schwartz and the fact that $\|P\|_\pi \leq 1$,
\[
\begin{aligned}
v^\transpose A_{NL}(\theta) v 
& \leq \beta \|Q^{\theta^1} - Q^{\theta^2} \| \|H^\theta (Q^{\theta^1} - Q^{\theta^2}) \|
\\
& \leq \beta \|Q^{\theta^1} - Q^{\theta^2} \|^2
\\
& = \beta v^\transpose A_L v.
\\
& = (\beta - 1) v^\transpose A_L v + v^\transpose A_L v.
\end{aligned}
\]
Rearranging the terms, and noting that $\Sigma_\psi = A_L$, the statement of the Lemma follows:
\begin{equation}
\begin{aligned}
v^\transpose A(\theta) v 
& = v^\transpose A_{NL}(\theta) v  - v^\transpose A_L v 
\\
& \leq (\beta - 1) v^\transpose \Sigma_\psi v
\end{aligned}
\label{e:Anegative}
\end{equation}
{
	Next, for fixed matrix $A(\theta)$ with eigenvalue-eigenvector pair $\lambda_A\in\mathbb C$, $v=a+bi \in\mathbb C^d$, we consider
	\[
	v^* A(\theta) v = (a^\transpose - b^\transpose i)A(\theta)(a + b i) = a^\transpose A(\theta)a + b^\transpose A(\theta)b +[a^\transpose A(\theta) b - b^\transpose A(\theta)a]i
	\]
	where $v^*$ denotes the conjugate transpose of $v$. With $v^* A(\theta) v = \lambda_A v^* v $, it follows that
	\[
	\Real\{\lambda_A\}  v^* v = \Real\{ v^* A(\theta) v\} = a^\transpose A(\theta)a + b^\transpose A(\theta)b
	\]
	Let $\lambda_\psi>0$ be the largest eigenvalue of $\Sigma_\psi$, by the inequality \eqref{e:Anegative}, the following relation holds
	\[
	\begin{aligned}
	\Real\{\lambda_A\}  v^* v &= a^\transpose A(\theta)a + b^\transpose A(\theta)b \\
	&\leq (\beta - 1)\lambda_\psi[ a^\transpose a + b^\transpose b]\\
	& = (\beta - 1)\lambda_\psi v^*v
	\end{aligned}
	\]
	Therefore, $\Real\{\lambda_A\}$ is negative and bounded above by $(\beta -1)\lambda_\psi$. 
	For the last part, $\|A(\theta)^{-1}\|_F$ is bounded using an inequality from \cite{cheng2014note}
	\begin{equation}
	\|A(\theta)^{-1}\|_F \leq \frac{\sqrt{d}}{|\det A(\theta)|} \Big(\frac{\|A(\theta)\|_F}{d-1} \Big)^{\frac{d-1}{2}}
	\end{equation}
	where $d$ is the dimension. Provided the bound over eigenvalues of $A(\theta)$ and compactness assumption of state space $\state$, there exists some constant $\ell_A$ such that
	\begin{equation}
	\label{eq:invBound}
	\|A(\theta)^{-1}\|_F \leq \frac{\sqrt{d}}{[(1-\beta)\lambda_\psi]^d}\Big(\frac{\ell_A}{d-1} \Big)^{\frac{d-1}{2}}
	\end{equation}
	The claim of uniform boundedness of $\{A(\theta)^{-1}:\theta\in \mathbb R^d \}$ then follows.
	\hfill$\blacksquare$
}
\bigskip
\bigskip

%

\paragraph*{Proof of \Lemma{t:bthetaLip}}
For any two parameter vectors $\theta^1, \theta^2 \in \Re^d$, we have:
\[
\begin{aligned}
\| b(\theta^1) - b(\theta^2) \|
& = \| \psi ( - F Q^{\theta^1} + F Q^{\theta^2} + Q^{\theta^1} - Q^{\theta^2} ) \|
\\
& \leq \| \psi ( F Q^{\theta^1} - F Q^{\theta^2} )  \| + \| \psi (Q^{\theta^1} - Q^{\theta^2} ) \|
\\
& \leq \beta \| \psi ( Q^{\theta^1} - Q^{\theta^2} )  \| + \| \psi ( Q^{\theta^1} - Q^{\theta^2} ) \|
\\
& \leq (1 + \beta) \|\psi\|^2 \|\theta^1 - \theta^2\|.
\end{aligned}
\]
\hfill$\blacksquare$

\bigskip
\bigskip

\begin{lemma}
	\label{t:stationaryNoise}
	Suppose that $f$ is a bounded function on $\state$ with zero mean.  Then the sequence $\{f(X_n) \}$ is  ODE-friendly,  with  $\{\epsy_k\}$ equal to zero,  and $ \clT_k = \hat f(X_k)$,  with $\hat f$ the solution to Poisson's equation with forcing function $f$; That is, $\haf$ satisfies \[
	\haf (x) = f(x) - \pi(f) + \Expect[\haf(X_{n+1}) | X_n = x] \,, x \in \state \,,
	\]
	where $\pi(f) \eqdef \Expect[f(X)] \,, X \sim \pi$.
	\qed
\end{lemma}

\begin{proof}
	Let $\hat f:\state\rightarrow \Re ^{d}$ solve Poisson's equation:
	\[
	\Expect[\hat f(X_{n+1}) - \hat f(X_n)|\mathcal F_n] = f(X_n)
	\]
	with the forcing function  $f(X_n)$ in Lemma \ref{t:stationaryNoise}. The following decomposition is obtained:
	\[
	\begin{aligned}
	f(X_n) &=  \Expect[\hat f(X_{n+1}) - \hat f(X_n)|\mathcal F_n] \\
	&= \underbrace{\Expect[\hat f(X_{n+1}) |\mathcal F_n] - \hat f(X_{n+1})}_{\text{Martingale difference}} + \underbrace{\hat f(X_{n+1}) - \hat f(X_n)}_{\text{telescoping}}
	\end{aligned}
	\]
	where $\hat f(X_n)$ is bounded since the forcing function $f(X_n)$ is bounded \cite{MT}.
\end{proof}
\ad{HAD not even introduced what Poisson's equation and forcing function means; Is this OK here?}
\ad{Lots of stuff here seem to be tooo technical for anyone to understand; Maybe move it to Appendix? Or we will give more details?}
\ad{$L_\infty$ undefined. This should be OK? }

\bigskip
\bigskip

\begin{lemma}
	\label{t:Econtinuous}
	Denote for $n\ge 0$,  
	$M_{\psi, \theta} ( {n+1}) = \psi(X_n)\clS_{\theta}\psi^\transpose(X_{n+1})$. There exists a constant $\ell_M<\infty$ such that, with probability one,
	\begin{equation*}
	\label{eq:noiseTerm}
	\| 	\Expect\big[M_{\psi, \theta^1} ( {n+1}) -M_{\psi, \theta^2} ( {n+1}) \mid \clF_n\big]  \|
	\leq \ell_M\|\theta^1 - \theta^2 \| 
	\end{equation*} 
	\qed
\end{lemma}

\begin{proof}
	For $\theta^1, \theta^2$, by Markov property with $X_n=x$, we rewrite the term as
	\[
	\begin{aligned}    &\hspace{-0.4in}\Expect\big[\psi(X_n)\clS_{\theta^1}\psi(X_{n+1})^\transpose|\mathcal F_n\big] - \Expect\big[\psi(X_n)\clS_{\theta^2}\psi(X_{n+1})^\transpose|\mathcal F_n\big]\\
	&\hspace{-0.2in}= \Expect\big[\psi(x)\clS_{\theta^1}\psi(X_{n+1})^\transpose|X_n=x\big] - \Expect\big[\psi(x)\clS_{\theta^2}\psi(X_{n+1})^\transpose|X_n=x\big]
	\end{aligned}
	\]
	
	By Assumption A3, the terminating cost function $c_s-1$ is in the span of basis. There exists some constant $\theta_{cs}$ such that
	\[
	c_s(y) \equiv \psi^\transpose(y) \theta_{cs} + 1,\quad \forall y\in \state.
	\]
	Denote $Z\sim\mathcal N(\psi(x), I)$ as the random variable with density $p_{Z}(\cdot)$, we have, for some $\ell_{E}>0$,
	\[
	\begin{aligned}
	&\hspace{-0.4in} \|\Expect\big[\psi(x)\clS_{\theta^1}\psi(X_{n+1})^\transpose|X_n=x\big] - \Expect\big[\psi(x)\clS_{\theta^2}\psi(X_{n+1})^\transpose|X_n=x\big]\| \\
	&\hspace{-0.6in}\leq \Expect\big[\|\psi(x)\psi^\transpose(X_{n+1})\| \cdot\big| \mathbb I\{\psi(X_{n+1})^\transpose\theta^1 \leq c_s(X_{n+1})\} - \mathbb I\{\psi(X_{n+1})^\transpose\theta^2 \leq c_s(X_{n+1})\}\big| |X_n=x\big] \\
	&\hspace{-0.6in}= \Expect\big[\|\psi(x)\psi^\transpose(X_{n+1})\| \cdot\big| \mathbb I\{\psi(X_{n+1})^\transpose(\theta^1-\theta_{cs}) \leq 1\} - \mathbb I\{\psi(X_{n+1})^\transpose(\theta^2-\theta_{cs}) \leq 1\}\big| |X_n=x\big]  \\
	&\hspace{-0.6in} = \int \|\psi(x)z^\transpose\|\cdot\ell(z|x) \cdot\big| \mathbb I\{z^\transpose(\theta^1-\theta_{cs}) \leq 1\} - \mathbb I\{z^\transpose(\theta^2-\theta_{cs}) \leq 1\}\big| p_Z(z)dz \\
	&\hspace{-0.6in}\leq \ell_E\int \big| \mathbb I\{z^\transpose(\theta^1-\theta_c) \leq 1\} - \mathbb I\{z^\transpose(\theta^2-\theta_c) \leq 1\}\big| p_Z(z)dz \\
	&\hspace{-0.6in}= \ell_E \Expect[\big| \mathbb I\{Z^\transpose(\theta^1-\theta_c) \leq 1\} - \mathbb I\{Z^\transpose(\theta^2-\theta_c) \leq 1\}\big| \big]
	\end{aligned}
	\]
	As a result, we only need to show that the expectation in the last line is Lipschitz continuous for a Gaussian random variable, for which a bounded derivative of $\Expect[\mathbb I\{Z^\transpose\theta \leq 1\}]$ w.r.t $\theta$ will suffice. For each $\theta$, $Z^\transpose\theta$ follows Gaussian distribution $\mathcal N(\theta^\transpose\psi(x), \theta^\transpose\theta)$. Denote $\Expect[\mathbb I\{Z^\transpose\theta \leq 1\}]$ as $m(\theta)$
	\[
	m(\theta)\eqdef\Expect[\mathbb I\{Z^\transpose\theta \leq 1\} \big] = \int_{-\infty}^1 \frac{1}{\sqrt{2\pi \theta^\transpose \theta}}\exp\Big(-\frac{1}{2\|\theta\|^2}(v-\theta^\transpose\psi(x))^2\Big)dv
	\]
	Applying the change of variable $u = \frac{1}{\sqrt{2\|\theta\|^2}}(v-\theta^\transpose\psi(x))$, we have
	\[
	\begin{aligned}
	m(\theta) &= \int_{-\infty}^{\frac{1}{\sqrt{2\| \theta\|^2}}(1-\theta^\transpose\psi(x))}\frac{1}{\sqrt{2\pi \|\theta\|^2}}\exp\big(-u^2\big)du \sqrt{2\|\theta\|^2}\\
	&= \frac{1}{\sqrt{\pi}}\int_{-\infty}^{\frac{1}{\sqrt{2\| \theta\|^2}}(1-\theta^\transpose\psi(x))} \exp \big(-u^2\big)du
	\end{aligned}
	\]
	Its derivative is
	\[
	\begin{aligned}
	\frac{d}{d\theta}m(\theta) = [\frac{1}{\sqrt{2\pi}} \frac{-\psi(x)\|\theta\| -(1-\psi(x)^\transpose\theta)\cdot\|\theta\|^{-1}\theta }{\|\theta\|^2}]\exp\big(-\frac{(1-\theta^\transpose\psi(x))^2}{2\|\theta\|^2}\big)
	\end{aligned}
	\]
	We observe that
	\begin{enumerate}
		\item If $\|\theta\|\rightarrow \infty$, the exponential term is bounded by 1, and the coefficient before exponential term goes to zero since
		\[
		\frac{\psi(x) }{\|\theta\|}+\frac{(1-\psi(x)^\transpose\theta) }{\|\theta\|^2} \rightarrow 0
		\]
		\item $m'(\theta)$ also vanishes as $\|\theta\|\rightarrow 0$ since
		\[
		\big[\frac{\psi(x) }{\|\theta\|}+\frac{(1-\psi(x)^\transpose\theta) }{\|\theta\|^2}\big]\exp\big(-\frac{(1-\theta^\transpose\psi(x))^2}{\sqrt{2\|\theta\|^2}}\big)\rightarrow 0
		\]
	\end{enumerate}	
	Since the state space $\state$ is compact, there exists a deterministic bound over $m'(\theta)$ that does not depend on $x$. As a result, the objective in Lemma \ref{t:Econtinuous} is Lipschitz continuous with some deterministic Lipschitz constant $\ell_M$.
\end{proof}

\bigskip
\bigskip

\begin{lemma}
	\label{t:poissonCont}
	There is a constant $B_Z$ such that the following holds:  For any family of zero mean functions  $\{f_\theta \}$ satisfying for some constants $B_F$, $\ell_F$,
	\[
	\begin{aligned}
	\sup_x \|f_\theta(x) \| 
	& \leq B_F,
	\\
	\sup_x  \| f_{\theta^1}(x) -f_{\theta^2}(x) \|
	&  \leq \ell_F \|\theta^1 - \theta^2 \|
	\end{aligned}
	\]
	for all $\theta,\theta^1,\theta^2$, then, there are solutions to corresponding Poisson's equation,  $\{\haf_\theta \}$, with zero mean, and satisfying
	\[
	\sup_x  \|  \haf_{\theta^1}(x) - \haf_{\theta^2}(x) \|
	\leq B_Z \ell_F \|\theta^1 - \theta^2 \|
	\] 
	\qed
\end{lemma}

\begin{proof}
	Let $P$ denote transition kernel of the Markov chain, we define the following fundamental kernel $\mathcal Z$
	\begin{equation}
	\label{eq:fund-kernel}
	\begin{aligned}
	\mathcal Z &\eqdef [I - P + \mathbbm 1\otimes \pi]^{-1} = \sum_{n=0}^\infty [P -\mathbbm 1\otimes\pi]^n = I + \sum_{n=1}^\infty [P^n - \mathbbm 1\otimes \pi]
	\end{aligned}
	\end{equation}
	Note that $\hat f_\theta = \mathcal Z f_\theta$ solves Poisson's equation:
	\[
	P\hat f_\theta = \hat f_\theta - f_\theta 
	\]
	Since $\mathcal Z:L_\infty\rightarrow L_\infty$ is a bounded linear operator, we have that for any $x\in \state$,
	\[
	\|h_{\theta^1}(x) -h_{\theta^2}(x)\|\leq \|h_{\theta^1} - h_{\theta^2}\|_{\infty} \leq \|\mathcal Z\|_\infty\|f_{\theta^1} - f_{\theta^2}\|_\infty \leq \ell_F\|\mathcal Z\|_\infty\|\theta^1 - \theta^2\|
	\]
	where $\|\mathcal Z\|_\infty$ denotes the induced operator norm.	
\end{proof}

\bigskip
\bigskip

\paragraph*{Proof of \Lemma{t:pre-ODEOS}}

Based on (\ref{e:QSNR2OSAdef}, \ref{e:QSNR2OSdef}), the two error sequences in  \eqref{e:SNR2OSlinearSA_preODE} are
\[
\begin{aligned}
\clE^\theta_{n+1} & = c(X_n) \psi(X_n) - b^*  +  \beta \clS^c_{\theta_n}\cstop(X_{n+1}) - \beta \barcstop(\theta_n) 
\\
\clE^A_{n+1} & = A_{n+1} - A(\theta_n)
\end{aligned}
\]
The argument proceeds by decomposing noise sequences into tractable terms, each of which is shown to be ODE-friendly by standard arguments based on solutions to Poisson's equation \cite{shwmak91}. We only present the treatment for $\clE^A_{n+1}$ here since same technique can be applied to $\clE^{A}_{n+1}\theta_{n}$ and $\clE_{n+1}^\theta$.

Denote for $n\ge 0$,  
\begin{equation*}
M_{\psi, \theta} ( {n+1}) = \psi(X_n)\clS_{\theta}\psi^\transpose(X_{n+1})
\end{equation*}
For noise sequence $\mathcal E_{n+1}^A$ in \eqref{e:SNR2OSlinearSA_preODE_A}, we have
\begin{equation}
\label{e:noiseA}
\begin{aligned}
\hspace{-.5em}
\clE^A_{n+1}   = A_{n+1} - A(\theta_n) & = \psi(X_n)\big[\beta \clS_{\theta_n}\psi(X_{n+1}) - \psi(X_{n})\big]^\transpose  -\Expect\big[\psi(X_n)[\beta \clS_{\theta_n}\psi(X_{n+1}) - \psi(X_{n})]^\transpose
\big]  \\
& = \tilA_{n+1}^1 + \beta \tilA_{n+1}^2 + \beta \tilA_{n+1}^3
\end{aligned}
\end{equation}
where
\[
\begin{aligned}
\tilA_{n+1}^1  &=  -   \psi(X_n)\psi^\transpose(X_{n}) + \Expect[\psi(X_n)\psi^\transpose(X_{n})]  
\\
\tilA_{n+1}^2  & =     M_{\psi, \theta_n} ( {n+1})  - \Expect[ M_{\psi, \theta_n} ( {n+1}) \mid \clF_n ] 
\\
\tilA_{n+1}^3  & =       \Expect[ M_{\psi, \theta_n} ( {n+1}) \mid  \clF_n ] - \Expect[ M_{\psi, \theta_n} ( {n+1})  ]     
\end{aligned}  
\]

\Lemma{t:stationaryNoise} implies that the sequence $\{\tilA^1_{n+1}\}$ is ODE-friendly since it is bounded over state space $\state$, and has zero mean. $\{\tilde A_{n+1}^2\}$ is ODE friendly as it is martingale difference sequence with bounded moments.

For $\tilde A_{n+1}^3$, let $\hat f_\theta:\Theta\times \state\rightarrow \Re ^{d\times d}$ solve Poisson's equation:
\[
\Expect[\hat f_\theta(X_{n+1}) - \hat f_\theta(X_n)|\mathcal F_n] = f_\theta(X_n)
\]
with forcing function  $f_\theta(X_n)\eqdef \tilde A^3_{n+1}$. The following representation is obtained:
\[
\begin{aligned}
f_{\theta_n}(X_n) =& \Expect[\hat f_{\theta_n}( X_{n+1}) - \hat f_{\theta_n}(X_n)|\mathcal F_n] \\
=& \Expect[\hat f_{\theta_n}(X_{n+1})|\mathcal F_n] - \hat f_{\theta_n}(X_{n+1}) + \hat f_{\theta_n}(X_{n+1}) - \hat f_{\theta_n}(X_n) \\
=& \underbrace{\Expect[\hat f_{\theta_n}(X_{n+1})|\mathcal F_n] - \hat f_{\theta_n}(X_{n+1})}_{\text{Martingale difference}} + \underbrace{\hat f_{\theta_{n+1}}( X_{n+1}) - \hat f_{\theta_n}( X_n)}_{\text{telescoping}}+ \underbrace{\hat f_{\theta_{n}}(X_{n+1})-  \hat f_{\theta_{n+1}}(X_{n+1})}_{\text{perturbation}}
\end{aligned}
\]
which admits the form of being ODE-friendly, provided the perturbation term $\epsy_n:= \hat f_{\theta_{n}}(X_{n+1})-  \hat f_{\theta_{n+1}}(X_{n+1})$ satisfies
\[
\sum_{k=1}^\infty \gamma_k \|\epsy_k\| < \infty,\quad a.s..
\]
Recall that Lemma \ref{t:Econtinuous} shows $\tilA_{n+1}^3$ is Lipschitz continuous with deterministic Lipschitz constant.  Combining this with Lemma \ref{t:poissonCont} shows that $\hat f_{\theta}(X_{n+1})$ is uniformly Lipschitz continuous w.r.t $\theta$ with Lipschitz constant $B_Z\ell_F$, which indicates 
\[
\begin{aligned}
\sum_{k=1}^\infty \gamma_k\|\epsy_k\| &= \sum_{k=1}^\infty \gamma_k\|\hat f_{\theta_{k}}(X_{k+1})- \hat f_{\theta_{k+1}}(X_{k+1})\| \leq  \sum_{k=1}^\infty \gamma_kB_Z\ell_F\|\theta_{k}-  \theta_{k+1}\| \\
&\leq B_Z\ell_F\sum_{k=1}^\infty \gamma_k\alpha_k\|\haA_k^{\dagger}\psi(X_k)(c(X_k) + \beta\min(c_s(X_{k+1}), Q^{\theta_k}(X_{k+1}))- Q^{\theta_k}(X_k))\|\\
\end{aligned}
\]
which is bounded by the boundedness assumption over $\{\theta_n\}$ and the fact that we use projected pseudo-inverse  of $\haA_{n+1}$.
\hfill$\blacksquare$

\bigskip
\bigskip

\paragraph*{Proof of \Lemma{t:odegain}}
Consider the update rules for $\theta_n,\haA_n$ in \eqref{e:QSNR2OSAdef} and \eqref{e:QSNR2OSdef}. When both update rules are viewed as over the faster time-scale (that is, with step-size sequence $\{\gamma_n\}$), they can be rewritten as \cite{bor08a}:
\begin{equation}
\begin{aligned}
\theta_{n+1} 
=& \theta_n + \gamma_{n+1}\big[o(1) \big] 
\\
\haA_{n+1} 
=& \haA_{n} + \gamma_{n+1}\big[ A(\theta_n) - \haA_n + \mathcal E^A_{n+1}\big]
\end{aligned}
\end{equation}
where the $o(1)$ term is:
\[
o(1) = - \frac{\alpha_{n+1}}{\gamma_{n+1}} \haA_n^{\dagger}\psi(X_n)\big[c(X_n) + \beta\min(c_s(X_{n+1}), Q^{\theta_n}(X_{n+1}))- Q^{\theta_n}(X_n)\big]
\]
It goes to zero provided stability assumption of $\{\theta_n\}$ and boundedness of $\haA_n^{\dagger}$.
It follows that $\{\theta_n, \haA_n\}$ converges {\em a.s.} to the internally chain transitive invariant set of the following ODE \cite{bor08a}
\[
\begin{aligned}
\dot{w}(t) &= 0\\
\dot{\mathcal A} (t) &= A(w_t) - \mathcal {A}(t)
\end{aligned}
\]
Which is $\{(\theta, A(\theta)):\theta\in\Re ^d \}$. In other words, $\|\hat{A}_n-A(\theta_n)\|$ converges to 0 {\em a.s.}. While the invertibility of $A(\theta)$ has been established in Lemma \ref{t:AthetaNegEig}.
\hfill$\blacksquare$

\bigskip
\bigskip

\paragraph*{Proof of \Lemma{t:odetheta}}
Within time interval $[s, s+T]$, consider the evolution of $\barw_{t_k}$ over slow time scale defined by $\{\alpha_k\}$, where $t_k$ denote the time of $k$-th update of $\theta$.
\[
\begin{aligned}
\barw_{t_{k+1}} = \barw_{t_k} - &\alpha_{k+1}\haA_{k+1}^{\dagger}\psi(X_{k})\Big[c(X_k) \\
&+ \beta\min(c_s(X_{k+1}),\barw_{t_{k}}^\transpose \psi(X_{k+1})) -\barw_{t_{k}}^\transpose \psi(X_{k})\Big]
\end{aligned}
\]
Write it in standard stochastic approximation form and replace $\haA_{k+1}^{\dagger}$ with $-A(\theta_k)^{-1} + o(1)$ by Lemma \ref{t:AthetaNegEig} and \ref{t:odegain}
\[
\begin{aligned}
\barw_{t_{k+1}} =\barw_{t_{k}} - &\alpha_{k+1}A^{-1}(\barw_{t_k})\Big[A(\barw_{t_k})\barw_{t_k} + \beta\barcstop(\barw_{t_k}) + b^* \\
&+ \mathcal E_{k+1}^A\barw_{t_k} +\mathcal E_{k+1}^\theta\Big] + o(\alpha_{k+1})
\end{aligned}
\]
With the noise terms being
\[
\begin{aligned}
\mathcal E_{k+1}^A\barw_{t_k} &= A_{k+1}\barw_{t_k} - A(\barw_{t_k})\barw_{t_k} \\
\mathcal E_{k+1}^\theta &= c(X_k) \psi(X_k) - b^*  +  \beta \cstop(X_{k+1})\ind{\{Q^{\theta_k}(X_{k+1}) > \cstop(X_{k+1}) \}} - \beta \barcstop(\theta_k)
\end{aligned}
\]
With them being shown ODE-friendly in Lemma \ref{t:pre-ODEOS}, the ODE approximation for $\barw_t$ follows under the stability assumption of $\theta_n$ \cite{tadic2003asymptotic}.
\hfill$\blacksquare$

\bigskip
\bigskip

\paragraph*{Proof of \Lemma{t:b_ode}}
For result in (i), the first relation holds since mapping $b(\cdot)$ is Lipschitz continuous. {The sub-sequential limit $w_t$ is the solution of the ODE in \eqref{eq:thetaODEsubseq} since the right hand side of \eqref{eq:thetaODEsubseq} is Lipschitz continuous. It is differentiable everywhere within $[0,T]$.

}

By definition in \eqref{e:btheta}, $b_t= b(w_t)$ can be written as follows:
\begin{equation}
\begin{aligned}
b_t
& =  - A(w_t) w_t - \beta \barcstop(w_t)
\\
& = -\Expect\big[\psi(X_n)  (\beta \min \big(c_s(X_{n+1}), Q^{w_t}(X_{n+1}) \big) - Q^{w_t}(X_n) ) \big]
\\
& = \Expect\big[\psi(X_n) c^{w_t}(X_n) \big]
\end{aligned}
\label{eq:costODEb}
\end{equation}
where, for any $w \in \Re^d$, $c^{w}$ is the cost function that solves the fixed point equation \eqref{e:QOSBell} with $\theta^*$ replaced by $w$:
\[
c^{w}(x) \eqdef -\Expect \big [\beta \min(c_s(X_{n+1}), Q^{w}(X_{n+1}))  \mid X_n = x \big] + Q^{w}(x), \qquad x \in \state
\]
The cost function $c^{w}$ is Lipschitz in $w$. Therefore it is Lipschitz continuous over $[0, T]$ and absolutely continuous over $[0,T]$ that has derivative almost everywhere.
Let $t_0$ be a point of differentiability for $c_t = c^{w_t}$. Both $w_t$ and $c_t(x)$ (for each $x \in \state$) are approximated by a line at this time point:
\[
\begin{aligned}
w_t &= w_{t_0} + (t-t_0)v^{w} + o(|t-t_0|)\\
c_t(x) &= c_{t_0}(x) + (t-t_0)v^{c}(x) + o(|t-t_0|), \quad t \sim t_0
\end{aligned}
\]
where $v^w, v^c(x)$ are the respective derivatives. The assertion 
\begin{equation}
v^c(x)= - \Expect[\beta \clS_{w_{t_0}}\psi(X_{n+1})^\transpose - \psi(X_{n})^\transpose \mid X_n = x] v^w
\label{e:vcx}
\end{equation}
will then imply the statement of the Lemma. 

Denote $L_t^c = c_{t_0}+(t-t_0)v^c$, $L_t^w = w_{t_0}+(t-t_0)v^w$. For each $t $, we have:
\[
\begin{aligned}
L^c_t(x) =
&-\Expect[\beta \min(c_s(X_{n+1}), Q^{L^w_t}(X_{n+1})) - Q^{L^w_t}(X_n) \mid X_n = x]
\\
=
& -\Expect[\beta \clS_{L^w_t}Q^{L^w_t}(X_{n+1}) + \beta \clS^c_{L^w_t}c_s(X_{n+1}) - Q^{L^w_t}(X_n) \mid X_n = x]
\\
\geq & -\Expect[\beta \clS_{w_{t_0}}Q^{L^w_t}(X_{n+1}) + \beta \clS^c_{w_{t_0}}c_s(X_{n+1}) - Q^{L^w_t}(X_n) \mid X_n = x]
\\
= & c_{t_0}(x) - (t-t_0)\Expect[\beta \clS_{w_{t_0}}\psi(X_{n+1})^\transpose - \psi(X_{n})^\transpose \mid X_n = x]v^w
\end{aligned}
\]
The above inequality is true for both $t > t_0$ and $t < t_0$. For $t \sim t_0$, it follows that the inequality becomes an equality, and therefore \eqref{e:vcx} holds. {The ODE \eqref{eq:costODE} for $b_t$ holds by replacing $w_t$ with $b_t$ in \eqref{eq:thetaODEsubseq} using the above derivative. }
\hfill$\blacksquare$

\end{document}